\documentclass[12pt]{article}

\usepackage{hyperref,url}
\usepackage{tikz-cd}
\usepackage{tikz}
\usepackage[round]{natbib}
\usepackage{amsmath}
\usepackage{amsbsy}
\usepackage{amscd}
\usepackage{amsopn}
\usepackage{amstext}
\usepackage{amsxtra}
\usepackage{amssymb}
\usepackage{amsthm,amsfonts}
\usepackage{graphicx}
\usepackage{epstopdf}
\usepackage{bbm}
\usepackage{epsfig}
\usepackage[title,titletoc]{appendix}
\usepackage{titling}

\newtheorem{theorem}{Theorem}
\newtheorem{proposition}[theorem]{Proposition}
\newtheorem{lemma}[theorem]{Lemma}

\theoremstyle{definition}
\newtheorem{definition}{Definition}

\newtheorem{assumption}{Assumption}
\newtheorem{remark}{Remark}

\def\beq{\begin{eqnarray}}
\def\eeq{\end{eqnarray}}
\def\beqq{\begin{eqnarray*}}
\def\eeqq{\end{eqnarray*}}
\def\beeq{\begin{eqnarray*}}
\def\eeeq{\end{eqnarray*}}
\def\be{\begin{equation}}
\def\ee{\end{equation}}
\newcommand{\bea}{\begin{eqnarray}}
\newcommand{\eea}{\end{eqnarray}}
\newcommand{\beaa}{\begin{eqnarray*}}
\newcommand{\eeaa}{\end{eqnarray*}}
\newcommand{\barr}{\begin{array}}
\newcommand{\earr}{\end{array}}

\newcommand{\benum}{\begin{enumerate}}
\newcommand{\eenum}{\end{enumerate}}

\newcommand{\norm}[1]{\left\lVert#1\right\rVert}

\def\f{\frac}

\def\Langle{\Bigl\langle}
\def\Rangle{\Bigr\rangle}

\def\BF={{\mathbb{F}}}
\def\BG={{\mathbb{G}}}
\def\BH={{\mathbb{H}}}
\def\bOne={{\bf 1}}

\def\CA{{\cal A}}

\def\CF{{\cal F}}

\def\CS{{\cal S}}
\def\CT{{\cal T}}

\def\dee{{\rm d}}

\def\diag{{{\rm diag}}}

\def\qed{\hfill$\sqcap\kern-7.0pt\hbox{$\sqcup$}$\\}

\def\BBC{\mathbb{C}}
\def\BBE{\mathbb{E}}
\def\BBP{\mathbb{P}}

\def\BBR{\mathbb{R}}

\def\BBZ{\mathbb{Z}}

\def\Ndft{{\rm Ndft}}
\def\omegamax{{\rm omegamax}}

\def\Bin{{\rm Bin}}

\def\One{{\bf 1}}

 \def\dee{\mathrm{d}}

\begin{document}

\title{\normalsize \bf 	COVID-19: ANALYTICS OF CONTAGION ON INHOMOGENEOUS RANDOM SOCIAL NETWORKS}
\author{\footnotesize T. R. HURD \vspace{-4ex}\\
\footnotesize \it Mathematics \& Statistics, McMaster University, 1280 Main St. West \vspace{-5ex}\\
\footnotesize \it Hamilton, Ontario, L8S 4L8, Canada \vspace{-5ex}\\
\footnotesize \it hurdt@mcmaster.ca}

%
\date{\footnotesize 
\today}
\maketitle

\begin{abstract}

Motivated by the need for robust models of the Covid-19 epidemic that adequately reflect the extreme heterogeneity of humans and society, this paper presents a novel framework that treats a population of $N$ individuals as an inhomogeneous random social network (IRSN). The nodes of the network represent individuals of different types and the edges represent significant social relationships.  An epidemic is pictured as a contagion process that develops day by day, triggered by a seed infection introduced into the population on day $0$. Individuals'  social behaviour and health status are assumed to vary randomly within each type, with probability distributions that vary with their type. A formulation and analysis is given for a SEIR (susceptible-exposed-infective-removed) network contagion model, considered as an agent based model, which focusses on the number of people of each type in each compartment each day. The main result is an analytical formula valid in the large $N$ limit for the stochastic state of the system on day $t$ in terms of the initial conditions. The formula involves only one-dimensional integration.  The model can be implemented numerically for any number of types by a deterministic algorithm that efficiently incorporates the discrete Fourier transform. While the paper focusses on fundamental properties rather than far ranging applications, a concluding discussion addresses a number of domains, notably public awareness, infectious disease research and public health policy, where the IRSN framework may provide unique insights.   \\



\end{abstract}
\renewcommand{\thefootnote}{\alph{footnote}}

Ê

Ê
Ê

\section{Introduction}
\label{sec:1}
Heterogeneity proliferates in human society at every level, and new types of mathematical modelling are needed to understand how these many layers of heterogeneity interweave and influence people's lives. The COVID-19 pandemic is a singularly far reaching and catastrophic event, and it will continue to negatively impact humanity for a long time to come. Layers of heterogeneity are especially relevant to a deep understanding of an infectious disease like COVID. Viral transmission may be through aerosols, droplets and fomites; the viral load may get absorbed by and do damage to a variety of tissues within the body; people's immune systems function in diverse ways. The virus itself may evolve into different forms. People are tremendously varied in their habits, their friendships, their living arrangements, their range of movements. These multifarious factors are all important to consider, and some will prove to be the most significant factors in determining where the disease will have its gravest damage, and the best actions to take to ameliorate this damage.

Network science has arisen in recent decades as the most helpful conceptual framework for handling potentially overwhelming complexity. Networks can provide the architecture and structure for agent based modelling of contagion, leading to the massive computer simulations such as those of \cite{Ferguson:2020nq} that have been used to develop a comprehensive picture of how such a disease may progress. As we shall show in this paper, network science can also provide shortcuts to dramatically speed up such computations, allowing us to quickly explore a vast array of alternative scenarios of the disease. 

This paper provides a novel network framework for society, called {\em  inhomogeneous random social networks} (IRSNs) and then models the propagation of an infectious disease like COVID-19 in such a society. It can be interpreted as an agent based contagion model, with the useful feature that an analytical shortcut is available for large-scale simulations of the  disease dynamics.   The framework starts with a so-called {\em inhomogeneous random graph (IRG)}, henceforth called the {\em social graph}, whose nodes represent people classified into a finite number of types, interconnected by edges representing their random social contacts. The people in this social network are  provided with random immunity buffers that measure their resistance to the disease and social contact links are labelled by random weights called exposures that quantify the viral load transmitted by infected individuals to their social contacts. Then, when a seed infection is introduced randomly into the population of susceptible individuals,  a sequence of contagion shocks will develop that will be modelled as iterations of a {\em cascade mapping} or  {\em cascade mechanism}.


The main contributions of this paper are:
\begin{enumerate}
  \item Introduction of the {\em inhomogeneous random social network (IRSN) framework} that provides a flexible and scalable architecture for describing a heterogeneous society of size $N$ with complex community structure. Individuals are classified by arbitrary types with random characteristics within each type. 
\item To develop infection cascade models for such networks based on a threshold mechanism for  transmission. This transmission mechanism can incorporate arbitrary dose-response functions, replacing over-simple transmission assumptions typically used in epidemic models. 
\item To develop large $N$ asymptotics for SEIR infection cascades in IRSN models, leading to Theorem \ref{Step1Prop} that provides explicit and efficiently computable recursive probabilistic formulas for the daily update of the state of the disease within the population. 
\item To show how the contagion analytics can be used to provide large scale investigations into potential policy interventions that one might invoke to mitigate or suppress the progress of the contagion. 
\item Overall, to provide a  {\em purely analytical toolkit} for networks with potentially thousands of different types of individuals, that can run on a laptop. The network framework is capable of providing much faster results, with a similar degree of accuracy, than is possible with large-scale agent based epidemic models sometimes used for informing health policy.  \end{enumerate}

Studying the spread of infectious diseases using the tools of network science has a substantial literature, reviewed for example in \cite{doi:10.1098/rsif.2005.0051} and \cite{Danonetal_NetworkEpidemics}. 
The book by \cite{Newman_Networks} provides a broad overview of networks in all areas of science, including applications to epidemic modelling, while 
\cite{PELLIS201558} explores current challenges in network epidemic models. Of particular interest  is the review of epidemic processes in complex networks by \cite{Pastor-Satorrasetal2015}: Many ingredients of the framework developed here can be traced to references described there. In particular, we see there that our model has its roots in the network cascade model of \cite{Watts02}, generalized to allow for random edge weights as in \cite{HurdGlee13}.

The IRSN model is presented here in a form exactly equivalent to a simple agent based model. This equivalence provides important motivation and justification of the underlying assumptions, and gives a vivid intuitive picture for interpreting the IRSN model. An important example of the intuition gained is a form of {\em selection bias} inherent in agent based models, and real epidemics, that we call {\em susceptibility bias}. Susceptibility bias, akin to Darwinian evolution, is the effect that  less resistant individuals tend to be infected earlier, leaving remaining susceptibles who tend to be more resistant than the original population. We will find that accounting for susceptibility bias presents a mathematical difficulty that our framework can partially, but not fully, solve.    In general, the IRSN framework lies between agent based models of the type developed by \cite{Ferguson:2020nq} and the literature on compartment ordinary differential equation models (ODEs) stemming from the pioneering work of \cite{KermMcKe27}. Full exploration of the conceptual links between these three distinct modelling frameworks is a promising avenue to deeper understanding of real world epidemics.

Section 2 of the paper introduces the essential structure of IRSNs and defines the SEIR infection cascade mechanism that characterizes the daily propagation of the disease on such networks. This section also explores the equivalent agent based model, and its heuristic properties.    Section 3 explores the large $N$  analytical properties of the IRSN model, leading to Theorem \ref{Step1Prop} that characterizes the infection cascade mapping on the first day. This result is extended to successive days by an additional mixing assumption,  providing a recursive characterization of the daily infection cascade mapping in the $N=\infty$ limit. Section 4 provides the ingredients for a numerical implementation of the SEIR cascade mapping that uses the discrete Fourier transform. It is shown that the flop count for computing a daily update is $O(M^2\times \Ndft)$ where $M$ is the number of types and $\Ndft$ is the number of lattice points in each one-dimensional integration. This is efficient enough that complex specifications of IRSN models can be explored quickly on a laptop.   Section 5 addresses the issue of calibrating IRSN models to real health and social data. In Section 6, we explore a simple illustration showing how the method can be used to understand potential policy interventions to protect the residents of a seniors' residential centre while a pandemic rages in the community outside. Finally, a concluding discussion addresses how this novel modelling framework can lead to improved understanding of epidemics by practitioners in several different domains.  \\

\noindent{\bf Notation:\ }
\begin{enumerate}
  \item For a positive integer $N$, $[N]$ denotes the set $\{1,2,\dots, N\}$. 
  \item For a  random variable $X$, its  cumulative distribution function (CDF), probability density function (PDF) and characteristic function (CF) will be denoted $F_X, \rho_X=F_X',$ and $ \hat f_X$ respectively. Note that  $ \hat f_X=\CF(\rho_X)$ where $\CF$ denotes the Fourier transform:
  \[ [\CF(f)](k)=\int^\infty_{-\infty}\ e^{ikx}\ f(x)\ \dee x\ ,k\in\BBR\ . \]  
  We also make use of the function $ \hat F_X:=\CF(F_X)$. 
  \item For any event $A$, $\One(A)$ denotes the indicator random variable, taking values in $\{0,1\}$.
  \item Landau's ``big O'' notation $f^{(N)}=O(N^\alpha)$ for some $\alpha\in\BBR$ is used for a sequence $f^{(N)}, N=1,2,\dots$ to mean that  $f^{(N)}N^{-\alpha}$ is bounded as $N\to\infty$.  
  \item The $L^2$ Hermitian inner product of two complex valued functions $f(x), g(x)$ on a domain $D$ is defined to be $\Langle f,g\Rangle_{L^2(D)}=\int_D\ f^*(x)\ g(x)\ \dee x$. The $L^2$ norm of a function $f(x)$ on a domain $D$ is defined to be $\norm{f}_{L^2(D)}=\langle f,f\rangle_{L^2(D)}^{1/2}$. 
\end{enumerate}  
%

 \section{SEIR Model on IRSNs } 
 This section provides the core modelling assumptions of the network epidemic framework, in the classic  {\em susceptible-exposed-infected-recovered} setting (see \cite{AndeMay92})  in which individuals progress through the stages of the disease, moving from compartment to compartment:
 \[ {\rm susceptible\ S}\rightarrow  { \rm exposed\  E}\rightarrow {\rm infective\  I}\rightarrow {\rm removed\  R}\ .\]    
The social network describes a population of individuals as {\em nodes} of a graph, whose {\em undirected edges} represent the existence of a significant social connection. Our network setting for the spread of an infectious disease has the following structure \begin{enumerate}
  \item The population is classified into a finite disjoint collection of ``types'' that represent people's important attributes, such as age, gender, living arrangement, profession, country and location.   
  \item Individuals within a type have random attributes drawn from type-dependent probability distributions.
  \item The network of social contacts, initially random, is taken to be constant during the epidemic. 
  \item The outbreak is monitored in discrete time, with a period $\Delta t$ assumed for convenience to be one day. At the start of the outbreak on day $0$, most of the population is susceptible (S),  but a small number of  individuals are exposed (E) or infective (I).
  \item Each day infective individuals pass on a random viral dose to their {\em infective contacts}, a  random subset of their social contacts. 
    \item A susceptible individual's state of health at the beginning of each day is represented by a random {\em immunity buffer}. During the day, they experience an accumulation of random viral doses through their infective contacts, and if the total viral load exceeds their buffer they become exposed, meaning infected but not yet infectious, and are moved into  compartment E. 
 \item Each day certain individuals move from  E to I, meaning they become infectious. Others move from I to removed (R) meaning that they either die or recover and are no longer infectious.
\end{enumerate}

This framework is a  kind of {\em agent based model} (ABM) that focusses solely on the occurrence of infective contacts, simpler than ABMs that also simulate the movements of individuals. Note that for each agent, their actual infection event is modelled as a threshold event that occurs if the total viral load received in a period of time $\Delta t=1$ day through multiple contacts exceeds their natural immunity.

The population with its social structure will be represented at any moment as an inhomogeneous random social network, or IRSN. An IRSN is the specification of a multidimensional random variable that captures two levels of structure. The {\em primary level}, called the {\em social graph}, is an undirected random graph with $N$ nodes labelled by a type classification, where each undirected edge represents the existence of a significant social connection, such as a family, collegial or friend relationship.  The {\em secondary layer} specifies the {\em infective contacts, mutual exposures and health} of people. 
Inhomogeneity in the IRSN model arises through classifying people by a finite number of types that can account for a wide range of attributes.     

It is important to note that the IRSN will be assumed to change over time in a prescribed fashion: the primary level remains constant, while the secondary layer varies stochastically each daily time step. The primary level is fixed because the calibration of the social graph is assumed to be based on studies such as \cite{MOssongPolymod2008} and  \cite{PremCookJit17} that studied contact data gathered  over many years prior to the outbreak. On the other hand, the secondary layer changes to reflect the stochastic nature of the pandemic on a daily scale. 
 
\subsection{Social Graph} The social graph is modelled as an undirected inhomogeneous random graph (IRG),  generalizing Erd\"os-Renyi random graphs, in which  edges are drawn independently between unordered pairs of nodes, not with equal likelihood but with likelihood that depends on their types. This class of random graph has its origins in \cite{ChungLu02} and has been studied in generality in \cite{BolSvaRio07} and the textbook by \cite{vdHofstad16}.  

\begin{assumption}\label{A1}[Social Graph] The primary layer of an IRSN, namely the social graph ${\rm IRG}(\BBP,\kappa,N)$,  is an inhomogeneous random graph with $N$ nodes labelled by $v\in [N]$. It can be defined by two collections $\CT,\CA$ of random variables: $T_v$ for $ v\in[N]$ and $A_{vw}$ for  $(v,w)\in [N]\times [N]$ .
\begin{enumerate}
\item Each node $v\in [N]$, representing a person, has type  $T_{v}$ drawn independently with probability $\BBP(T)$ from a finite list of {\em types} $[M]$ of cardinality $M\ge 1$.  Note that $\sum_{T\in[M]}\BBP(T)=1$.
\item Each undirected edge $(v,w)\in [N]\times [N]$ corresponds to a non-zero entry of the symmetric random {\em adjacency matrix} $A$. For each pair $(v,w)$,  $A_{vw}=A_{wv}$ is the indicator for $w$ to be (significantly) socially connected to $v$.  Conditioned on the collection of all types $\{T_{v}\}$, the collection of edge indicators  $\{A_{vw}\}$ is an independent family of Bernoulli random variables with probabilities 
  \be\label{PMK} \BBP[A_{vw}=1\mid T_v=T,T_w=T']=(N-1)^{-1}\kappa(T,T')\One(v\ne w)\ . \ee

  \end{enumerate}
  \end{assumption}

It is an important observation that the sequence of  ${\rm IRG}(\BBP,\kappa,N)$ with the same $\BBP,\kappa$ and varying size $N$ have uniform probabilistic characteristics that tend to a central limit as $N\to\infty$.  In particular,  the {\em probability mapping kernel} $\kappa$, the symmetric matrix that determines the likelihood that two people $v,w$ of the given types have a social connection, is divided by $N-1$ to ensure this uniformity and sparseness of the graph for large $N$.  For consistency we require that $N-1\ge \max_{T,T'}\kappa(T,T')$.

\subsection{Infective Contacts, Viral Exposures and Immunity Buffers}\label{healthsection} The relevant health attributes of all people are summarized by an independent collection of multivariate random variables,  conditioned on the social graph.
\begin{definition}\label{healthbuffer} 
\begin{enumerate}
\item The {\em infective contact indicator pair} between $w$ and $v$ is a pair of Bernoulli variables $(\zeta_{vw},\zeta_{wv})$. $\zeta_{vw}A_{vw}=1$ means that the social relationship between $v$ and $w$ leads to a {\em close infective contact} on a given day, such that when $v$ is infectious a viral dose will be transmitted to $w$. 
  \item 
The {\em potential viral exposure pair} between $w$ and $v$ is a pair $(\Omega_{vw},\Omega_{wv})$ of positive values: $\Omega_{vw}$ represents the viral load transmitted from $v$ to $w$ should $v,w$ have a single {\em infective contact}, and if $v$ is infective. 
  \item 
 The {\em immunity buffer} $\Delta_v$ of node $v$  is a non-negative value that represents the resistance of that person to the virus. 
\end{enumerate}
\end{definition}



\begin{assumption}\label{A2}[Infective Contacts, Viral Exposures and Immunity Buffers]
The secondary layer of an IRSN, the collection of infective contacts, potential exposures and immunity buffers $z_{vw}, \Omega_{vw}, \Delta_v$   are non-negative random variables that are chosen to be independent of $\{A_{vw}\}$,  conditioned on $\{T_v\}$.

\begin{enumerate}
\item For each edge $(v,w)$, $(\zeta_{vw},\zeta_{wv})$ is a bivariate Bernoulli random variable. Conditioned on $T_v=T,T_w=T'$, $\zeta_{vw}=1$  with probability $z(T,T')$. 
 \item For each edge $(v,w)$, $(\Omega_{vw}, \Omega_{wv})$ is a bivariate random variable. Conditioned on $T_v=T,T_w=T'$, $\Omega_{vw}$ has  a continuous marginal density  $\rho_{\Omega}(x\mid T, T')$ supported on $\BBR_+$ and associated distribution functions $F_\Omega(\cdot|T,T'), \hat f_\Omega(\cdot|T,T'), \hat F_\Omega(\cdot|T,T')$.
\item For each individual $v$, $\Delta_v$ conditioned on $T_v=T\in[M]$ has a continuous  density $\rho_\Delta(x|T)= F'_{\Delta}(x|T)$ supported on $\BBR_+$. Thus the cumulative distribution function (CDF) is  \be
 F_{\Delta}(x|T):=\BBP(\Delta_v\le x \mid T)=\int^x_0\rho_\Delta(y|T)\dee y\quad,\forall x\in [0,\infty)\cup\{\infty\}\ .
 \label{BarBufCDF}
\ee
We also record the characteristic function (CF) $\hat f_{\Delta}(\cdot|T)=\CF(\rho_\Delta(\cdot|T))$ and $\hat F_{\Delta}(\cdot|T)=\CF(F_\Delta(\cdot|T))$.
\end{enumerate}
\end{assumption}

In summary, a IRSN of finite size $N$ representing the population of $N$ individuals amounts to a collection of random variables $\{T,A,\zeta,\Omega,\Delta\}$ satisfying Assumptions 1 and 2.

\subsection{Infection Transmission and the Epidemic Trigger}\label{triggersection} 
Infection transmission is a stochastic process that we idealize here as proceeding in discrete time with a period taken for convenience to be one day. This time scale can be thought to correspond to the length of time a transmitted viral load remains active within the body. The most important factors in determining the probability that a susceptible individual becomes infected in a day are the total viral load they accumulate during that day and their immunity buffer. We adopt a {\em threshold infection assumption}, as described for example in  \cite{Pastor-Satorrasetal2015}[Ch X].

We assume  that the random social graph determined by $\{T,A\}$  is chosen at time $t=0$ and remains fixed for the duration of the contagion process. On the other hand, $\{\zeta,\Omega,\Delta\}$ form a conditionally IID sequence of  multivariate random variables that are drawn daily. Thus,  only the secondary layer of the IRSN changes over time. 

Consider a typical day starting at time $t, t\in \BBZ_+$, at which time the compartments S,E,I,R are assumed to be a union over $T\in[M]$ of disjoint random subsets  $S(t|T), E(t|T),I(t|T),R(t|T)$ of the node set $[N]$. The initial compartments, and the possible compartment changes each day are determined by the following rules:
\begin{assumption}\label{A3}[Initial Trigger and Transmission]
\begin{enumerate}
\item The {\em epidemic trigger} at the beginning of day $t=0$ randomly assigns each type $T$ individual to one of the {\em compartments} $S,E,I,R$ independently with  probabilities $s(0|T)=1- e(0|T)-i(0|T),e(0|T),i(0|T),r(0|T)=0$. This determines the initial compartments $S(0)=[N]\setminus(E(0)\cup I(0)), E(0),I(0), R(0)=\emptyset$; these compartments are partitioned by types: $S(0)=\cup_T S(0|T)$, etc.
\item Each day $t\ge 0$, a new collection $\{\zeta^{(t)},\Omega^{(t)},\Delta^{(t)}\}$ of random variables are sampled satisfying Assumption 2. 
\item For each successive day $t\ge 0$, the transmissions  
from S to E, E to I and I to R are determined by the following  SEIR transmission assumptions:
\begin{enumerate}
\item Each $v\in S(t|T)$ will be exposed and moved to $E(t+1|T)$ if 
\be\label{infection} \sum_w \One(w\in  I(t))\zeta^{(t)}_{wv}A_{wv}\Omega^{(t)}_{wv}\ge \Delta^{(t)}_v\ . 
\ee 
 \item Each $v\in E(t|T)$ becomes infectious and moves to $I(t+1|T)$ independently with probability $\gamma(T)\in[0,1]$.
  \item Each $v\in I(t|T)$ is removed to $R(t+1|T)$ independently with probability $\beta(T)=\beta_d(T)+\beta_r(T)\in[0,1]$, where $\beta_b,\beta_r$ are the probabilities of death and recovery respectively. 
\end{enumerate}\end{enumerate}

\end{assumption}

Note that from the above assumptions,  $z(T',T)$ represents the conditional probability on a given day that $w$ and $v$ have an {\em infective contact}, given that they have a social contact and $w\in I(t|T'), v\in S(t|T)$; these are entries of an  $M\times M$ possibly non-symmetric matrix.  

As will be discussed in detail in Section \ref{doseresponse}, the threshold infection assumption captured in \eqref{infection} can be directly interpreted as a {\em dose-response model} as reviewed in \cite{Haas_Doseresponse_2014}. 

%

\subsection{IRSN Agent Based Simulation}\label{secABM} Assumptions 1,2,3 for an infection contagion cascade of Tmax days duration on a finite social network of $N$ people can be realized by the following algorithm for the IRSN-ABM,  a simple agent based model.  
\begin{enumerate}
  \item[Step 0] Initialize the primary level random variables $T_v,A_{vw}$ according to Assumptions \ref{A1}. Set $t=0$ and assign each node $v\in[N]$ independently to one of the compartments $S(0),E(0),I(0),R(0)$ according to the initial probabilities $s(0|T_v), e(0|T_v),i(0|T_v),$ $ r(0|T_v)$ as in Assumption \ref{A3}. 
  \item[Step 1] While $t< \mbox{Tmax}$:\begin{enumerate}
  \item Update the secondary random variables: For each $w\in I(t)$ and $v\in S(t)$,  generate  $\zeta^{(t)}_{wv},\Omega^{(t)}_{wv},\Delta^{(t)}_v,$ according to Assumptions \ref{A2} .
  \item  Exposure:  For $v\in S(t)$, if $\sum_{w\in I(t)} \zeta^{(t)}_{wv}A_{wv}\Omega^{(t)}_{wv}\ge \Delta^{(t)}_v$, move $v$ to $E(t+1)$, otherwise keep $v\in S(t+1)$.
  \item For $v\in E(t)$, independently move $v$ to $I(t+1)$ with probability $\gamma(T_v)$, otherwise keep $v\in E(t+1)$.

  \item For $v\in I(t)$, independently move $v$ to $R(t+1)$ with probability $\beta(T_v)$, otherwise keep $v\in I(t+1)$.
\item Increment $t=t+1$ and repeat Step 1.
\end{enumerate}
\end{enumerate}

Each simulation of the model leads to the collection of random compartments\\ $S(t|T),E(t|T),I(t|T),R(t|T)$ with fractional sizes $$s(t|T)=N^{-1}|S(t|T)|, e(t|T)=N^{-1}|E(t|T)|,i(t|T)=N^{-1}|I(t|T)|, r(t|T)=N^{-1}|R(t|T)|$$
 for days $t=0,1,\dots, \mbox{Tmax}$ and types $T\in[M]$. 
 
 \begin{remark}\label{rem1} The above specification for the IRSN-ABM  is one of many natural possibilities. In particular, choosing to freeze the social graph to remain constant, while making the secondary layer change unpredictably every day is a strong immunological assumption that is open to debate.  For example, one might propose an alternative assumption that the secondary layer exhibits serial correlation, or more strongly, remains constant day by day. 
 
 {\em Susceptibility bias} refers to a type of selection bias, akin to Darwinian evolution, that in a heterogeneous population where individuals have slowly varying innate characteristics, the less resistant individuals tend to succumb to the disease earlier than more resistant individuals, and consequently the susceptible population becomes more resilient over time.  Random variables such as the social graph that remain constant  lead to {\em susceptibility bias}, while making random characteristics serially independent reduces susceptibility bias. The IRSN-ABM will have some susceptibility bias arising from the constancy of the social graph because highly connected individuals of a given type will tend to receive more infectious shocks than less connected individuals of the same type.  \end{remark}

\section{Analytical Asymptotics of the IRSN model}

The IRSN framework just introduced specifies the joint distributions of the random variables $\{T,A,\zeta,\Omega,\Delta\}$ and the random compartments $S(0),E(0),I(0),R(0)$, thereby providing a compact stochastic representation of the state of a network of $N$ individuals at the moment an outbreak is triggered. The same distributional data defines a sequence of random networks with varying $N$. 

The main objective is to study the dependence on $t$ of the size of the random compartments  $S(t|T), E(t|T),I(t|T),R(t|T)$. It is important that we consider relationships between finite $N$ networks and the asymptotic limit $N\to\infty$. To this end, for each $N$ we define the fractional expected sizes to be
\begin{eqnarray}
s^{(N)}(t|T) & := & \frac1{N}\BBE^{(N)}\Bigl[\sum_{v\in[N]}\One(v\in S(t|T))\Bigr] \\
e^{(N)}(t|T) & := & \frac1{N}\BBE^{(N)}\Bigl[\sum_{v\in[N]}\One(v\in E(t|T))\Bigr] \\
i^{(N)}(t|T) & := & \frac1{N}\BBE^{(N)}\Bigl[\sum_{v\in[N]}\One(v\in I(t|T))\Bigr] \\
r^{(N)}(t|T) & := & \frac1{N}\BBE^{(N)}\Bigl[\sum_{v\in[N]}\One(v\in R(t|T))\Bigr] \ .
\end{eqnarray}
By permutation symmetry, $s^{(N)}(t|T) =\BBP^{(N)}(1\in S(t|T)) $ etc. Throughout the remainder of the paper, the quantities  $s(t|T), e(t|T), i(t|T), r(t|T)$ without superscript $(N)$ denote the large $N$ limiting values.

The most important result of the paper will be $N\to\infty$ asymptotic recursion formulas  mapping the quantities $s(t|T), e(t|T), i(t|T), r(t|T)$  from day $t$ to $t+1$ for $t\ge 0$, subject to specified initial conditions for $t=0$. This system of equations is  a discrete dynamical system on a simplex defined by relations $\BBP(T)=s(t|T)+ e(t|T)+ i(t|T)+ r(t|T)$ for each $T\in[M]$,  lying within the hypercube $[0,1]^{4M}$. The mapping generating this dynamics will be called the {\em infection cascade mapping}.  

 \subsection{Degree Distribution of the Social Graph} \label{degree}
 
The distribution of the number of social contacts of nodes in IRGs, in other words their {\em social degree distribution}, has a natural Poisson mixture structure in the large $N$ limit.   By permutation symmetry,  one only needs to consider individual $1$ with arbitrary type $T_1=T$, whose social degree  is defined as  $\dee_1=\sum_{w=2}^N A_{w1}$, a sum of conditionally IID random variables.  Since $e^{ikA_{w1}}=1+A_{w1}(e^{ik}-1)$, each term has the identical conditional characteristic function (CF)
 \be \BBE^{(N)}[e^{ikA_{w1}}\mid T_1=T]=\sum_{T'\in[M]}\BBP(T')\left(1+(N-1)^{-1}\kappa(T',T)(e^{ik}-1)\right)\ .
 \ee
 The conditional CF of $\dee_1$ is the $N-1$ power of this function, and  can be written
\be \BBE^{(N)}[e^{ik\dee_1}\mid T]=\left[1+\frac1{N-1}\sum_{T'}\BBP(T')\kappa(T',T)(e^{ik}-1)+O(N^{-2})\right]^{N-1}\ ,\label{degreeCF}
\ee
to display its asymptotic structure as $N\to\infty$.  
\begin{proposition}\label{Prop1} The characteristic function of the  social degree $\dee_v$ of an individual $v$, conditioned on its type $T\in[M]$, is $2\pi$-periodic   on $\BBR$ and has the $N\to\infty$ limiting behaviour:    
\beq \label{prop1}\hat f^{(N)}(k\mid T)&=&\hat f(k\mid T)\   \left(1+O(N^{-1}) \right)\ ,\\
\hat f(k\mid T)&:=& \exp\left[\lambda(T)(e^{ik}-1)
  \right]\ ,
\eeq
  where $\lambda(T)=\sum_{T'}\lambda(T', T)$ with $\lambda(T', T):=\BBP(T')\kappa(T',T)$. Here,  convergence of the logarithm of \eqref{prop1} is in $L^2([0,2\pi])$.
  \end{proposition}

\begin{proof}[Proof of Proposition \ref{Prop1}]  The proof is immediate by applying the following Lemma \ref{limitlemma} to the logarithm of \eqref{degreeCF}, 
with $N-1=y^{-1}$ and $ g(k,y) = \sum_{T'\in[M]}\BBP(T')\Bigl[\kappa(T',T)(e^{ik}-1)\Bigr]. $
\end{proof}

 \begin{lemma}\label{limitlemma} Let $I$ be any hyperinterval in $\BBR^d$ and $\bar y>0$. Suppose $g(x,y):I\times[0,\bar y]\to\BBC$ is a bivariate function such that  $g(\cdot,y), \partial_y g(\cdot,y), \partial^2_y g(\cdot,y)$ are pointwise bounded and in $L^2(I)$ for each value $y\in[0,\bar y]$. Then 
 \[ \lim_{y\to 0} \norm{\frac1{y}\log(1+yg(x,y))- g(x,0)}_{L^2}=O(y)\ .\]
\end{lemma}

  \begin{proof}[Proof of Lemma \ref{limitlemma}] Under the assumptions, one can show directly that $f(x,y):=\log(1+yg(x,y))]- yg(x,0)$ satisfies $\lim_{y\to 0}f(x,y)=\lim_{y\to 0}\partial_yf(x,y)=0$ and hence by Taylor's remainder theorem
  \[ f(x,y)=\int^y_0(y-v)  \partial^2_yf(x,v)\dee v\ . \]
  One can also show that $\partial^2_yf(x,v)$ is in $L^2(I)$ for each value $v\in[0,\bar y]$ provided $\bar y>0$ is small enough. Then, by Fubini's Theorem, for $y\in[0,\bar y]$
  \[ \norm{\log(1+yg(x,y))]- g(x,0)}^2\le (\int^y_0(y-v)\dee v)^2\max_{v\in[0,\bar y]}\norm{\partial^2_yf(x,v)}^2\le M y^4\]
  for some constant $M$, from which the result follows.
\end{proof}

Proposition \ref{Prop1} tells us that for different values of $T$, the conditional social degree converges in distribution to a  Poisson random variable with mean parameter $\lambda(T)=\sum_{T'} \lambda(T,T')$. Now, recall that a {\em finite mixture} of a collection of probability distribution functions is the probability distribution formed by a convex combination. Thus  the asymptotic {\em unconditional}  social degree distribution of any individual is a finite mixture with characteristic function:
\be  \hat f(k)= \sum_{T\in[M]}\BBP(T) \hat f(k\mid T)\ .\ee
Each mixture component has a Poisson distribution with Poisson parameters
$\lambda(T)$
 and the mixing variable is the individual-type $T$ with mixing weight $\BBP(T)$.


 \subsection{The First Infection Cascade Step}
The most important quantity on day $1$ is the {\em exposure probability} ${\rm EP}(1|T)$ for a type $T$  individual $v$ that is susceptible on day $0$ to become exposed on day 1. For a finite network of size $N$, by permutation symmetry, we can take $v=1$ and the required conditional probability can be expressed as ${\rm EP}^{(N)}(1|T):=\BBP^{(N)}(1\in E(1|T)|1\in S(0|T))$. By our assumptions, in particular \eqref{infection},  this is 
 \be\label{EP} {\rm EP}^{(N)}(1|T)=\BBP^{(N)}(\Delta^{(0)}_1\le V^{(0)}_1|1\in S(0|T))= \BBP^{(N)}(\Delta^{(0)}_1\le V^{(0)}_1|T_1=T)\ee
 where $V^{(0)}_1$, the total viral load received by $1$, is the sum of viral shocks from $w\ne 1$ to $1$
 \begin{eqnarray}
V^{(0)}_1&=& \sum_{w\ne 1}V^{(0)}_{w1}\ ;\nonumber\\
V^{(0)}_{w1}&=&\sum_{T'\in[M]} A_{w1}\zeta^{(0)}_{w1}\One(w\in I(0|T'))\Omega^{(0)}_{w1} \ .
\end{eqnarray}
As studied in \cite{HurdGlee13}, threshold probabilities such as \eqref{EP} are efficiently computable via characteristic functions. Assuming that $X,Y$ are independent non-negative random variables such that $X$ has a density $\rho_X(x)$ and the CDF $F_Y$ of $Y$ is continuous, and letting these functions have Fourier transforms $\hat f_X:=\CF(\rho_X),\ \hat F_Y:= \CF(F_Y)$, then by the Parseval Identity
\be\label{Parseval} \BBP[Y\le X]=\int^\infty_0\ \rho_X(x)F_Y(x)dx= \f1{2\pi} \int_{-\infty}^\infty \hat F_Y(k)\hat f_X(-k)\ \dee k .
\ee  

 Note that conditioned on $T_1=T$, the shocks $V^{(0)}_{w1}$ for all $w\ne 1$ are independent and identically distributed (IID).   Since $\sum_{T'}A_{w1} \zeta^{(0)}_{w1}\One(w\in I(0|T'))$ is a Bernoulli random variable that is independent of $\Omega^{(0)}_{w1}$,
 \[\exp[{ik\sum_{T'}A_{w1} \zeta^{(0)}_{w1}\One(w\in I(0|T'))\Omega^{(0)}_{w1}}]=1+\sum_{T'}A_{w1}\zeta^{(0)}_{w1}\One(w\in i(0|T'))(e^{ik \Omega^{(0)}_{w1}}-1)\ ,
 \] and hence for any $w\ne 1$ the characteristic function of the shock $V^{(0)}_{w1}$  conditioned on the type $T_1=T$ is given for finite $N$ by  \beq \BBE^{(N)}[e^{ikV^{(0)}_{w1}}\mid T]&=&\nonumber\\
 &&\hspace{-1in}1\ +\ \sum_{T'}\BBE^{(N)}[A_{w1} \zeta^{(0)}_{w1}\One(w\in I(0|T'))\mid  T_1=T]\ \BBE^{(N)}[e^{ik\Omega^{(0)}_{w1}}-1\mid  T_1=T,T_w=T']\nonumber \\
 &=&1+ \sum_{T'}\frac{\kappa(T',T)z(T',T)i(0|T')}{N-1} \Bigl( \hat f_\Omega(k\mid T',T)-1\Bigr) \label{shock1}\ . \eeq
 By the independence of viral shocks, the total viral load $V^{(0)}_1$ has CF\begin{equation}
\label{ }
\BBE^{(N)}[e^{ikV^{(0)}_{1}}\mid T]=\left(1+\sum_{T'}\frac{\kappa(T',T)z(T',T)i(0|T')}{N-1} \Bigl( \hat f_\Omega(k\mid T',T)-1\Bigr) \right)^{N-1}
\end{equation}
 The desired large $N$ approximation is uniform in $k\in\BBR$, and follows by the argument proving Proposition \ref{Prop1}:
  \beq \BBE^{(N)}[e^{ikV^{(0)}_{1}}\mid T]&=&\hat f^{(0)}_{V}(k\mid T)(1+O(N^{-1}))\ ,\\
 \hat f^{(0)}_{V}(k\mid T) &:=&\exp \left(\sum_{T'}\kappa(T',T)z(T',T)i(0|T') \Bigl( \hat f_\Omega(k\mid T',T)-1\Bigr)\ \right) \ .\label{cascade2}\eeq
 
The expected fractional number of type $T$ newly exposed individuals will be
$ {\rm EP}^{(N)}(1|T)s(0|T) $. By combining \eqref{EP} and \eqref{Parseval} with $X=V^{(0)}_1$ and $Y=\Delta_1^{(0)}$, and applying the dominated convergence theorem,  we have
\[ \lim_{N\to\infty} {\rm EP}^{(N)}(1|T)= \f1{2\pi}  \int^\infty_{-\infty} \hat F_\Delta(k|T)\ \hat f^{(0)}_V(-k|T)\ \dee k\ .
\]
The expected fractional number of type $T$ exposed individuals that become infective will be $\gamma(T)e(0|T)$ and expected fractional number of type $T$ infective individuals that are removed will be $\beta(T)i(0|T)$. Putting these pieces together, one obtains the main result. 

 \begin{theorem}\label{Step1Prop} Consider the sequence of IRSNs for all $N$, satisfying all the assumptions in Section 2. Then for each $N$, the fractional expected compartment sizes on day $1$ are 
\begin{eqnarray}
s^{(N)}(1|T) & = &(1-{\rm EP}^{(N)}(1|T))s(0|T)  \\
e^{(N)}(1|T) & = &(1-\gamma(T))e(0|T)+ {\rm EP}^{(N)}(1|T)s(0|T)\\
i^{(N)}(1|T) & = & (1-\beta(T))i(0|T)+\gamma(T) e(0|T)\\
r^{(N)}(1|T) & = & r(0|T)+\beta(T)i(0|T)\ .
\end{eqnarray}
The type $T$ exposure probability on day $1$ is uniformly approximated as $N\to\infty$:\begin{eqnarray}
{\rm EP}^{(N)}(1|T) & = & {\rm EP}(1|T)) (1+O(N^{-1}))\ ,\\
{\rm EP}(1|T) & = &\f1{2\pi}  \int^\infty_{-\infty} \hat F_\Delta(k|T)\ \hat f^{(0)}_V(-k|T)\ \dee k\ , \label{EP_step1}
\end{eqnarray}
where $\hat F_\Delta=\CF(F_\Delta)$ is  given by \eqref{BarBufCDF} and $\hat f^{(0)}_V=\CF(\rho^{(0)}_V)$ is given by \be
 \hat f^{(0)}_{V}(k\mid T) :=\exp \left(\sum_{T'}\kappa(T',T)z(T',T)i(0|T') \Bigl( \hat f_\Omega(k\mid T',T)-1\Bigr)\ \right) \ .
 \ee
\end{theorem}

\subsection{The Mixed Infection Cascade Mapping}
\label{IncremMap}

As discussed in Section \ref{secABM}, the IRSN model exhibits susceptibility bias. Due to the constancy over time of the social graph, highly connected individuals tend to be infected earlier than less connected individuals, and hence the average connectivity of susceptibles decreases over time. This implies that the assumptions underlying Theorem \ref{Step1Prop} do not hold for the infection cascade mapping on subsequent days $t>0$. Indeed, we have not been able to generalize Theorem \ref{Step1Prop} to cope with susceptibility bias for $t>0$.  

Instead, we propose to depart from the IRSN model of Section 2.3 by introducing an additional randomization called {\em mixing} that eliminates the susceptibility bias and ensures that Theorem  \ref{Step1Prop} holds  for $t>0$.  The required form of conditional independence is achieved by introducing for each $t>0$ a random reassignment of the labels $S,E,I,R$ for each type $T$, consistent with the total fractions of nodes in each subcompartment. Specifically, for each $t\ge 0$ we replace Step 1(e) of the IRSN agent based simulation of Section 2.4 by the following:\\\\
    Step 1'(e): Reassign each node $v\in[N]$ independently to one of the compartments $S(t+1),E(t+1),I(t+1),R(t+1)$ according to the  probabilities $s(t+1|T_v), e(t+1|T_v),i(t+1|T_v),$ $ r(t+1|T_v)$. Increment $t=t+1$ and repeat Step 1. \\

Since the proposed mixing is inconsistent with any true agent based model, we call this a {\em pseudo-agent based model}, the IRSN P-ABM. Under model specifications where susceptibility bias of the ABM is small, the large $N$ limit of the P-ABM will closely mimic the corresponding ABM. Under the IRSN P-ABM, the form given by Theorem \ref{Step1Prop} continues to apply for subsequent time steps, and we are justified in proposing the following mapping as a consistent network model for infectious disease spread.  

\bigskip\noindent
{\bf Mixed Infection Cascade Mapping:\ } Consider the limit $N=\infty$ of the sequence of IRSN P-ABMs for all $N$, satisfying all the assumptions in Section \ref{secABM}, with the modified Step 1'(e) just discussed. Then on day $t\ge 1$, 
 \begin{enumerate}
  \item The type $T$ exposure probability  is 
\be\label{poissonEP_t1}{\rm EP}(t|T) = \frac1{2\pi}\int^\infty_{-\infty}  \hat F_\Delta(k|T)\ \hat f^{(t-1)}_V(-k|T)\ \dee k\ , 
\ee
  \item  The transmitted viral shock has PDF  $\rho^{(t-1)}_V(\cdot|T)=\CF^{-1}(\hat f^{(t-1)}_V(\cdot|T))$ with \be
   \hat f^{(t-1)}_{V}(k\mid T) =\exp \left(\sum_{T'}\kappa(T',T)z(T',T)i(t-1|T') \Bigl( \hat f_\Omega(k\mid T',T)-1\Bigr)\ \right) \ .\label{cascade_t}\ee 
 \item   The fractional expected compartment sizes  are 
\begin{eqnarray}
\label{CM1}s(t|T) & = &(1-{\rm EP}(t|T)s(t-1|T)  \\
\label{CM2}e(t|T) & = &(1-\gamma(T))e(t-1|T)+ {\rm EP}(t|T)s(t-1|T)\\
\label{CM3}i(t|T) & = & (1-\beta(T))i(t-1|T)+\gamma(T) e(t-1|T)\\
r(t|T) & = & r(t-1|T)+\beta(T)i(t-1|T)\ .
\label{CM4}\end{eqnarray}

\end{enumerate}

\subsection{Dose-Response Model of Transmission}\label{doseresponse} To obtain a more specific threshold model of transmission, leading to a better understanding of the immune buffers and exposures, this section develops the idea of {\em dose-response}, as discussed in e.g. \cite{Haas_Doseresponse_2014}, as a model of viral transmission. In a simple dose-response model for airborne disease transmission, each  viral dose transmitted from an infective to a susceptible host is assumed to be carried by a very large number $\Omega$ of airborne particles, thought of as either aerosol or droplet. These particles settle on tissues within the host, where each is assumed to have an independent identical small chance $\alpha$ to cause the host to become exposed. The probability that exposure occurs is therefore
\be\label{doseresponse2} P_{exposure}=\sum_{n=1}^N \Bin(N,\alpha, n)\sim 1- e^{-\alpha \Omega} 
\ee
where $\Bin(N,\alpha, \cdot)$ denotes the values of a binomial distribution, and the approximation is the Demoivre-Laplace limit theorem. 

There are many reasons why \eqref{doseresponse2} is oversimple, and it is common to replace it by a more general dose-response relationship
\be\label{doseresponse3} P_{exposure}=F( \Omega) 
\ee
for an increasing function with $F(0)=0, F(\infty)=1$.

 We can view this general dose-response as a threshold model. The parameter $\alpha=\alpha_T$, or the specific function $F=F(\cdot|T)$ can be assumed to depend strongly on the host's type $T$. We should also assume that $\Omega$ is random, depending on all of the infecting individuals and the host's type.  If we assume a type $T$ susceptible's buffer $\Delta_v$ is distributed with CDF $F(\cdot|T)$ and is independent of the  total exposure $\Omega$ which is given as a  random sum of infectious exposures $\sum^K_{k=1} \Omega_{w_kv}$ then, conditioned on the exposures $\Omega_{w_kv} $, the probability of $v$ being exposed will be
\be {\rm EP}(t|T, \Omega)=F(\sum^K_{k=1} \Omega_{w_kv}|T)=\BBP[\Delta_v\le  \sum^K_{k=1} \Omega_{w_kv} | T,\Omega_{w_kv} ]\ .
\ee
consistent with \eqref{EP}.
Taking an expectation over $\Omega$ leads to the  probability of $v$ being exposed  after day $t$, conditioned on $v\in S(t-1|T)$. This will be ${\rm EP}(t|T))s(t-1|T)$ where
\beq {\rm EP}(t|T)&=&\BBP[\Delta_v\le \sum_{w\in[N]} A_{wv}\zeta^{(t-1)}_{wv}\One(w\in I(t-1))\Omega^{(t-1)}_{wv} |v\in S(t-1|T) ]\\
&=&\BBE\Biggl[F\Bigl(\sum_{w\in[N]} A_{wv}\zeta^{(t-1)}_{wv}\One(w\in I(t-1))\Omega^{(t-1)}_{wv}\Bigr)|T_v=T\Biggr]\eeq

For the simplest dose-response \eqref{doseresponse2}, these expectations factorize, and with the help of Lemma 2 one finds 
\beq
 {\rm EP}(t|T)&=&1-\prod_{w\ne v}\BBE\Bigl[1+A_{wv}\zeta^{(t-1)}_{wv}\sum_{T'}\One(w\in I(t-1|T')(e^{-\alpha_T\Omega^{(t-1)}_{wv} }-1)|T_v=T\Bigr]
\\
&&\hspace{-.5in}=\ 1-\Bigl[1+\sum_{T'}\frac{\kappa(T',T)}{N-1}z(T',T)i(t-1|T')\BBE[e^{-\alpha_T\Omega^{(t-1)}_{wv} }-1|T_v=T, T_w=T']\Bigr]^{N-1}\\
&\sim&1- \exp[-\sum_{T'}\tau(T',T)\kappa(T',T)z(T',T)i(t-1|T')]\label{DRexp}
\eeq
Here $\tau(T',T):=\BBE[e^{-\alpha_T\Omega^{(t-1)}_{wv} }-1|T_v=T, T_w=T']=\hat f_\Omega(i\alpha_T|T,T')-1$ is the probability that $v$ becomes exposed from a single viral dose from a type $T'$ infective. 

  When \eqref{CM1} and  \eqref{DRexp} are combined, we obtain 
\be\label{euler}  s(t|T)=\exp[-\sum_{T'}\tau(T',T)\kappa(T',T)z(T',T)i(t-1|T')] s(t-1|T)\ .\ee
This can be recognized as the Euler discretization with time step $\Delta t= 1$ day of the vector-valued ordinary differential equation at the heart of multi-type compartment epidemic models:
\be \frac{s(t|T)}{\dee t} = -\sum_{T'} \alpha(T,T')i(t|T')s(t|T)\ee
with the transmission parameter given by $\alpha(T,T') =\delta t^{-1}\tau(T',T)\kappa(T',T)z(T',T)$. Thus, combining  \eqref{DRexp} with the steps leading to Theorem 2 provides a direct derivation of the classic SEIR ODE model from a particular specification of a more fundamental agent based contagion model.

\section{Discrete Fourier Transform Implementation}  
\label{numerics} 
The core of the numerical implementation of the mixed infection cascade mapping will be to approximate the integral \eqref{poissonEP_t1} for ${\rm EP}(t|T)$ using the Discrete Fourier Transform (DFT). The DFT works most effectively on a grid of nonnegative integers we denote by $[\Ndft]:=\{0,1,2,\dots, \Ndft-1\}$ whose log-size $\log_2(\Ndft)$ is an integer chosen to compromise between precision and computational efficiency. All immunity buffers will be taken to have integer values on $[\Ndft]$ that represent multiples of a unit of viral dose. The exposures will have values on a smaller grid $\{0,1,2,\dots, \omegamax-1\}$. To avoid the {\em aliasing problem} familiar in applications of the DFT, we assume $\Ndft$ is sufficiently large compared to $\omegamax$ so that $\BBP(\sum_{w\in[N]\setminus 1} A_{w1}\zeta^{(t)}_{w1}\Omega^{(t)}_{w1}\One(w\in I(0))\ge \Ndft)$ is a negligible probability when node $1$ has any possible type $T_1$.

Thus we assume that the  PDF and CDF $\rho_X, F_X$ of any continuous random variable $X$ can be replaced by  dimension $\Ndft$ probability vectors with components $\rho_X(x), F_X(x), x\in[\Ndft]$.    
The characteristic function $\hat f_X$  is then replaced by the DFT of $\rho_X$, $\hat f_X:=\CF(\rho_X)$, defined for each $k\in [\Ndft]$ by
\[ \hat f_X(k)=\sum_{x\in[\Ndft]} e^{2\pi i kx/\Ndft} \rho_X(x)\]
The DFT is an invertible linear operator (in fact an isometry under the Euclidean metric) on $\BBC^{\Ndft}$; the inverse DFT $\rho_X=\CF^{-1}(\hat f_X)$ is given by
\[ \rho_X(x)=\Ndft^{-1}\sum_{k\in[\Ndft]} e^{-2\pi i kx/\Ndft} \hat f_X(k)\ .\]
Given two independent positive random variables $X,Y$ with values in $[\Ndft]$, one then has the identities
\[\BBP(X\ge Y)=\sum_{y\in[\Ndft]} F_Y(y)\ \rho_X(y)=\frac1{\Ndft}\sum _{k\in[\Ndft]}\hat F_Y(k) \hat f_X(-k)
\]
where $ \hat F_Y= \CF(F_Y)$.

Based on these identities, with the grid $[\Ndft]$ set in this way, we  can implement the mixed infection cascade mapping given by equations \eqref{CM1}-\eqref{CM4} of Section \ref{IncremMap} with equations \eqref{poissonEP_t1} and \eqref{cascade_t} replaced by
\beq
{\rm EP}(t|T)&=&\frac1{\Ndft}\sum _{k\in[\Ndft]}\hat F_\Delta(k|T) \hat f^{(t-1)}_V(-k|T)\\
\hat f^{(t-1)}_V(k|T)&=& \exp[\sum_{T'\in[M]} R(k,T,T') i(t-1|T')]\label{hatF}\\
R(k,T,T')&=&\kappa(T,T')z(T',T)(\hat f_\Omega(k|T,T')-1)\ .
\eeq

%

One sees that for a single day $t$, the computational complexity of the algorithm is dominated by \eqref{hatF} which amounts to  $O(\Ndft\times M^2)$ flops for the complex matrix-vector multiplication, followed by  $\Ndft\times M$ complex exponentiations. Memory usage is dominated by storing the constant matrix $R$ with $\Ndft\times M^2$ components. Since $\Ndft=2^{10}$ is a typical value, there is clearly no difficulty in computing the general model with several thousand types  on an ordinary laptop.

\section{Calibrating IRSNs} \label{data}

This section  addresses  implementation of the infection cascade model on IRSNs, and its generalizations, for a real world network of $\hat N$ individuals.  The central issue is to construct a sequence of IRSNs of size $N$ increasing to infinity, that is statistically consistent with the real world network when $N=\hat N$. Then the statistical model for $N=\infty$ can be subjected to epidemic triggers with any initial infection probabilities, and the resultant infection cascade analytics developed in Section 3 will yield the chronology of the epidemic, and measures of the resilience of the real world network. 

The type of network data available to policy makers varies widely from one health jurisdiction to another. Here we imagine a minimal dataset for $\hat N=\sum_{T\in[M]} \hat N_T$ individuals classified into $M$ types labelled by $T\in[M]$,  where $\hat N_T$ denotes the number of individuals of type $T$. Individual types, and the population sizes $\hat N_T$, will be assumed not to change over the data sampling period. As a first estimation step, we choose the {\em empirical type distribution}:
\[ \widehat\BBP(T)=\frac{\hat N_T}{\hat N}\ .\]
Typically, this vector is determined by census data. 

\subsection{Social Contact Matrix}

Next suppose for illustration that the interconnectivity, exposures and health statistics of the real network have been observed for an extended period prior to the epidemic. In particular, social connectivity has been observed, and  edges having the meaning of a ``significant social contact'' are drawn between any ordered pair $(v,w)$ of individuals if the average  daily contact time of individual $w$ to individual $v$ exceeds a specified threshold.

 Let the {\em social contact matrix} $\hat E=( \hat E_{T,T'})_{T,T'\in[M]} $ represent the expected total daily number of significant $T$ to $T'$ social contacts in the given population. Such matrices have been studied in great depth, for many countries and communities, and are available in public databases such as  \cite{PremCookJit17}. Following the discussion of Section \eqref{degree}, the average number of $T'$ contacts per type $T$ individual,  $\hat E_{T,T'}/\hat N_T$, is matched to  the conditional mean $\hat\lambda(T', T):=\hat\BBP(T')\widehat\kappa(T',T)$  to identify the  {\em empirical connection kernel}
\be\label{kappa1}  \widehat\kappa(T,T')=
  \frac{\hat E_{T,T'}(\hat N-1)}{\hat N_T\hat N_{T'}}   \  .\ee

Theoretically, social contact matrices can be constructed as a very large sum over {\em settings} $s\in\CS$ that represent the different places people meet, see \cite{mistry2020inferring}. Each setting $s$ is assumed to involve a finite number of people, with an equal likelihood $z(s)\in[0,1]$ of a contact
between any pair. Let $\phi^{(s)}=( \phi^{(s)}_{T,})_{T\in[M]} \in \BBZ_+^{M}$ denote the column vector counting the number of individuals of each type in the setting $s$. The construction then amounts to representing the matrix $\hat E$ by the weighted  sum
\be
\hat E=\sum_{s\in\CS} z(s)\left(\phi^{(s)}*\phi^{(s)\top}-\diag(\phi^{(s)})\right)\label{contact}
\ee
whose $T,T'$ component is $$\hat E_{TT'}=\sum_{s\in\CS} z(s)\left(\phi^{(s)}(T)\phi^{(s)}(T')-\phi^{(s)}(T)\delta_{TT'}\right)\ $$ where $\delta_{TT'}$ denotes the Kronecker delta. 
Note that $*$ in \eqref{contact} represents an outer-product of vectors, that in general yields a rectangular matrix. This sum over settings can be disaggregated into different types of settings, such as school, hospital and workplace, leading to contact matrices within subcommunities.

\subsection{Buffers and Exposures}
Recall from the previous section that exposures are assumed to take values on the integer grid $\{0,1,2,\dots,\omegamax\}$ for some moderately large integer $\omegamax$, where $1$ represents a choice of a unit dose. It should be supposed that $\Omega_e$ is observed for a certain random sample of directed $T\to T'$ edges  $e$. It is then reasonable to infer empirical distributions $\rho_\Omega(\cdot,T,T')$  from a parametric family of discrete distributions on $\{0,1,2,\dots,\omegamax\}$  that match the sample means and variances $\hat\mu_\Omega(T,T'),\hat\sigma^2_\Omega(T,T')$. In a similar way, $ \hat\mu_\Delta(T),\hat\sigma^2_\Delta(T)$ can be estimated from a random sample of observed values of the buffer variable for type $T$ nodes.

Gamma distributions on $\BBR_+$, parametrized by the shape parameter $k>0$ and scale parameter $\theta>0$, form a particularly nice family suitable for theoretical studies of the IRSN framework. Of particular interest are the exponential distributions with  $k=1$, due to their ``memoriless'' property. When $\Delta$ is exponential, the dose-response curve leads to the assumed serial independence of infection arising from successive viral doses. As shown in Section 3.4, this specification  leads to the Euler approximation of a compartment ODE model. Thus when $\Delta$ is exponential, the IRSN model should closely mirror properties of the ODE model. Heuristics seem to suggest however that the true dose-response function, i.e. the CDF of $\Delta$, is better taken to be ``S''-shaped with $k>1$. There is little literature on the statistics of transmitted viral doses from which to infer properties of the exposures $\Omega$, although  COVID-19 pharyngal swab test studies such as \cite{Jones_viralloads_2020} provide some insight. Interestingly, that study suggests that COVID-19 viral loads measured for individuals that tested positive may to be very fat-tailed.  For our present exploratory purposes,  gamma distributions may reasonably be used for both $\Delta$ and $\Omega$.  

 
\subsection{Infective Contact Parameters} Finally, one needs to identify the fractions $z(\cdot,\cdot)$ of  social contacts that are close infective contacts. First note that these fractions can be directly targeted by health policy, and will therefore be manipulated and changed dramatically during the pandemic. Since the fraction $z(T,T')$ applies where the type $T$ individual is infectious and the type $T'$ individual is susceptible,  lowering this parameter by restricting type $T$ behaviour will directly reduce the exposure of type $T'$ individuals. 

Since there is little fundamental theory that informs the choices of buffer and exposure distributions, and the $z$ parameters are the direct target of policy interventions, the most practical approach to $z$ is to chose benchmark values $\hat z(T,T')$ that provide a strong match for the contagion dynamics in the earliest stages of the pandemic. Thus, for example,  the observed effective $R$-naught value, $R_0$, should be used to calibrate the benchmark values.  Subsequently, the value of the $z$ parameters will need to be adjusted to account for proposed and actual changes in policy.

\section{Illustrative Example: Seniors' Residential Centre}
The purpose of this example is to provide an easy-to-visualize context for the IRSN framework, namely the setting of a seniors residence with 100 residents (type $T=1$), 50 trained staff workers (type $T=2$) within a town of total population $N_0= 10000$. We also consider the same IRSN specification scaled up by an integer multiplier $N=jN_0$. 

In anticipation of an oncoming contagion, the workers have been trained to high standards of hygiene and care and the residents (who are elderly but healthy) have been instructed in social-distancing and hygiene. The townspeople (``outsiders'', with type $T=3$) on the other hand have only average ability to social distance, and so the contagion hits the town before the centre. The goal of this example is to investigate the vulnerability of the residential centre to internal contagion starting in the outside town. 

The benchmark network parameters are given in Table \ref{KVRL_params}, together with numerical implementation parameters $\omegamax=60, \Ndft=256$. The buffers $\Delta$ and exposures $\Omega$ are all taken to be Gamma-distributed with shape parameter $k=3$, and means $\mu$ and standard deviations $\mu/\sqrt{3}$  that depend on type. 

\begin{table}[htp]
\caption{Benchmark Parameters: Note that $\kappa(T',T)\BBP(T)$ is the expected daily number of social contacts of a type $T'$ individual to type $T$ individuals.  }
\small{\begin{center}
\begin{tabular}{|c||c|c|c|}
\hline
&Resident $T=1$&Worker $T=2$&Outsider $T=3$\\
\hline\hline
$\beta(T)$&0.09&0.09&0.09\\ \hline
$\gamma(T)$&0.3&0.3&0.3\\ \hline
$z(T)$&0.20&0.20&0.20\\ \hline

$\BBP(T)$&0.01&0.005&0.985\\\hline
 $\kappa(1,T)\BBP(T)$& 4 &   5 & 0\\  
 $\kappa(2,T)\BBP(T)$&   10  &  5  &  4\\
  $\kappa(3,T)\BBP(T)$&   0 &   0.0203 &   20\\
  \hline
    $\mu_\Omega(1,T)$&7&7&7\\
    $\mu_\Omega(2,T)$&4&4&4\\
 $\mu_\Omega(3,T)$&6&6&6\\
   \hline
 $\mu_\Delta(T)$&20&30&30\\
\hline

   \end{tabular}
\end{center}}
\label{KVRL_params}
\end{table}

The upper left plot of Figure \ref{Fig1} shows the daily exposed, infective and removed fractions for the three types, in the benchmark SEIR model without further policy interventions, plotted from the day that the number of exposed outsiders exceeds 1\% of the population. We see that the contagion starts in the outside community, but rapidly invades the centre, resulting in similar infection rates, with a time delay. One can interpret the result as overlapping sub-epidemics: the first hits the outside community, while a second and third hit the residence workers and residents about 16 and 22 days later, respectively. One can see that the strategy failed for two reasons: first, the contagion was allowed to gain a foothold in the centre and infect a resident; second, the  hygiene within the centre was not adequate to contain the resulting seed infection. 
\begin{figure}[ht]
\centering
\includegraphics[width=0.5\columnwidth]{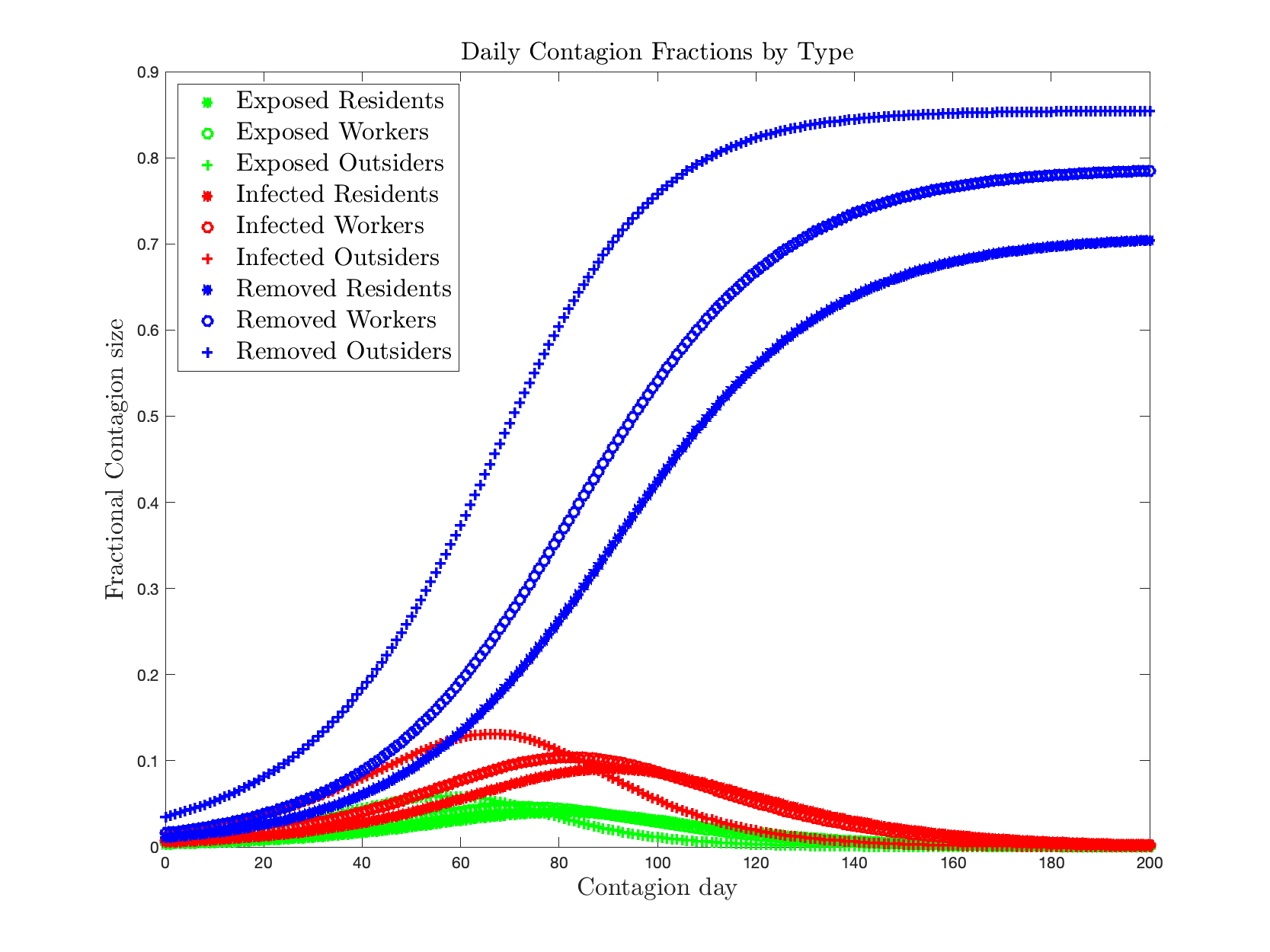}\includegraphics[width=0.5\columnwidth]{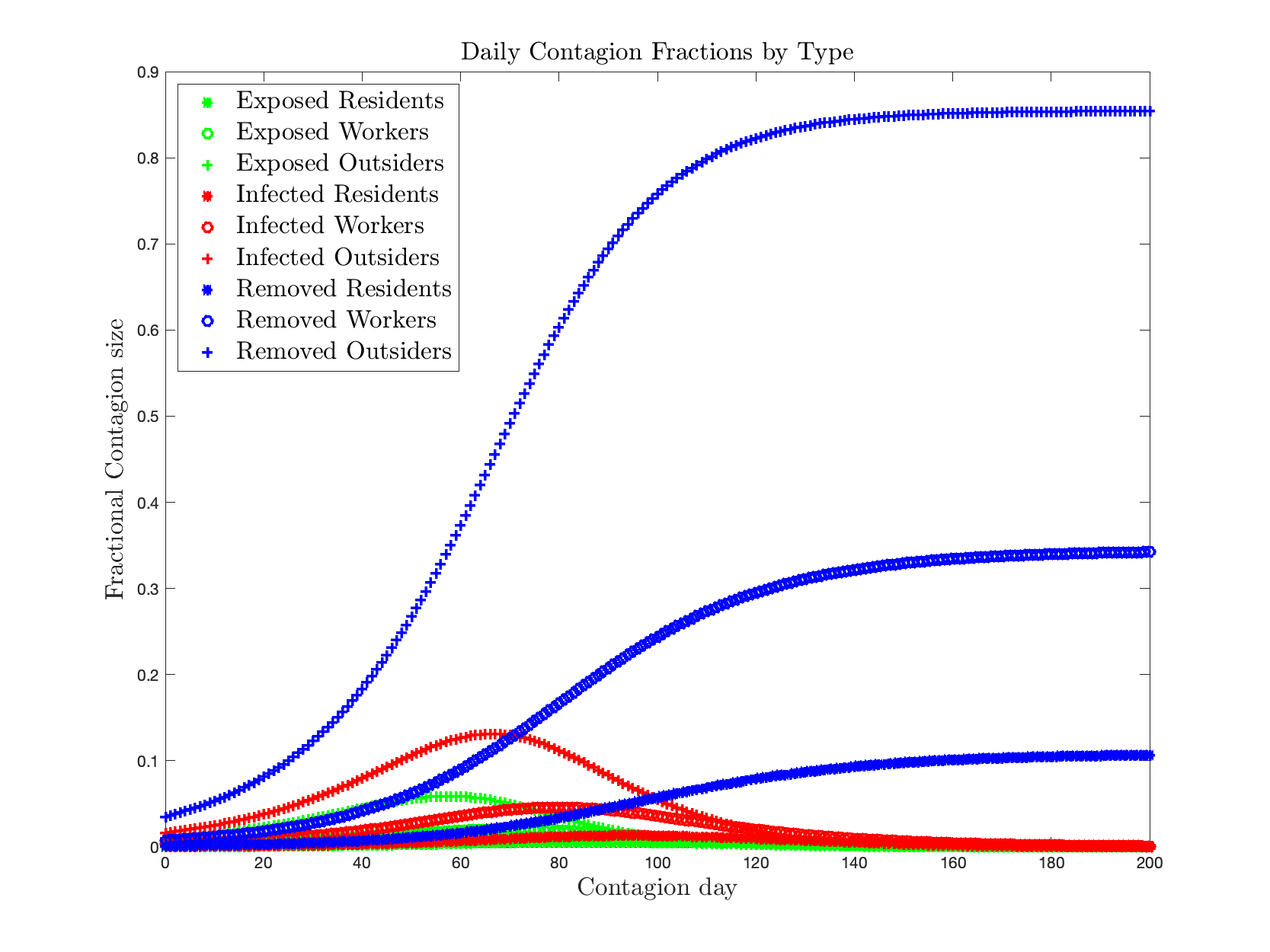}\\
\includegraphics[width=0.5\columnwidth]{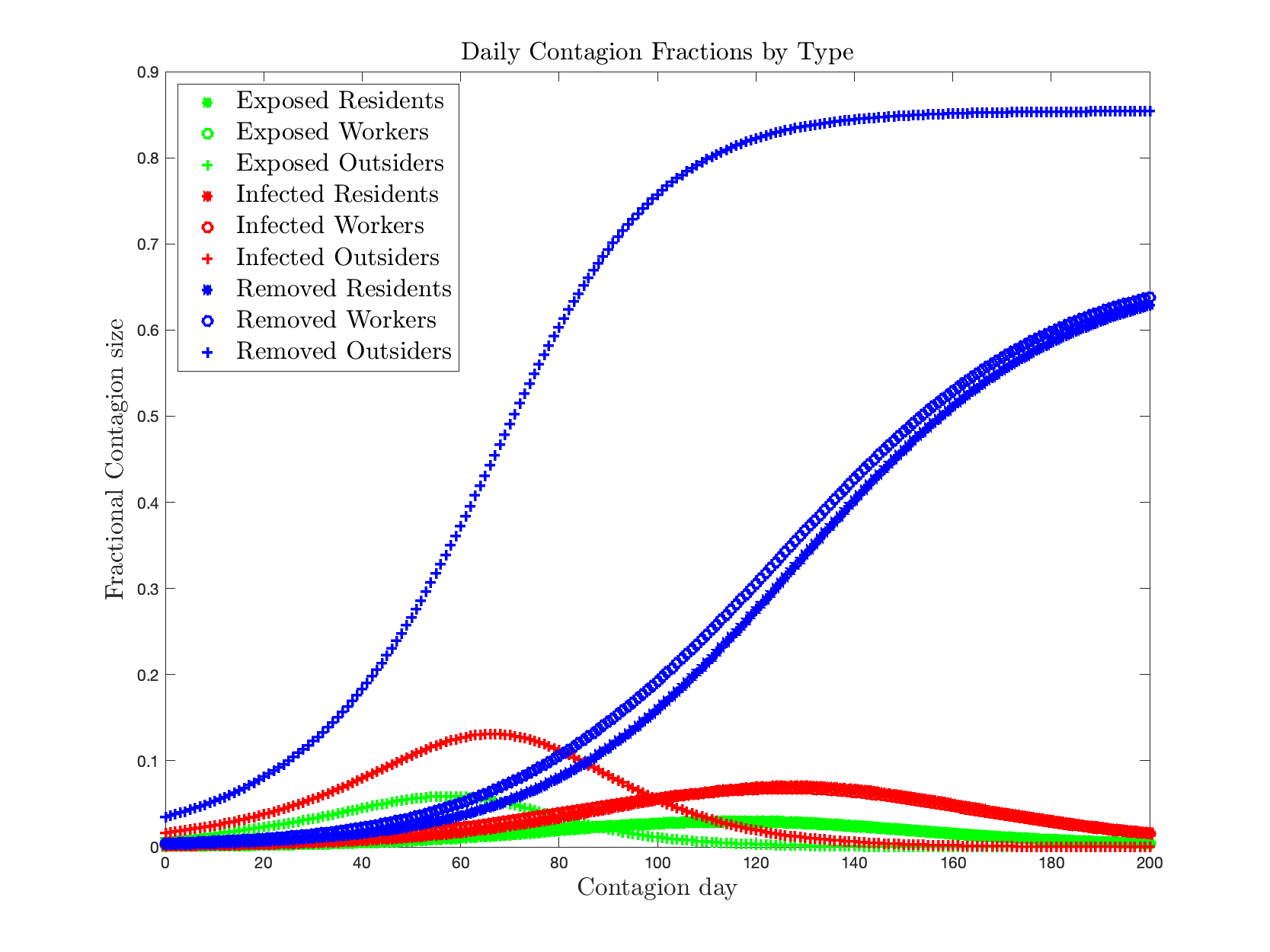}\includegraphics[width=0.5\columnwidth]{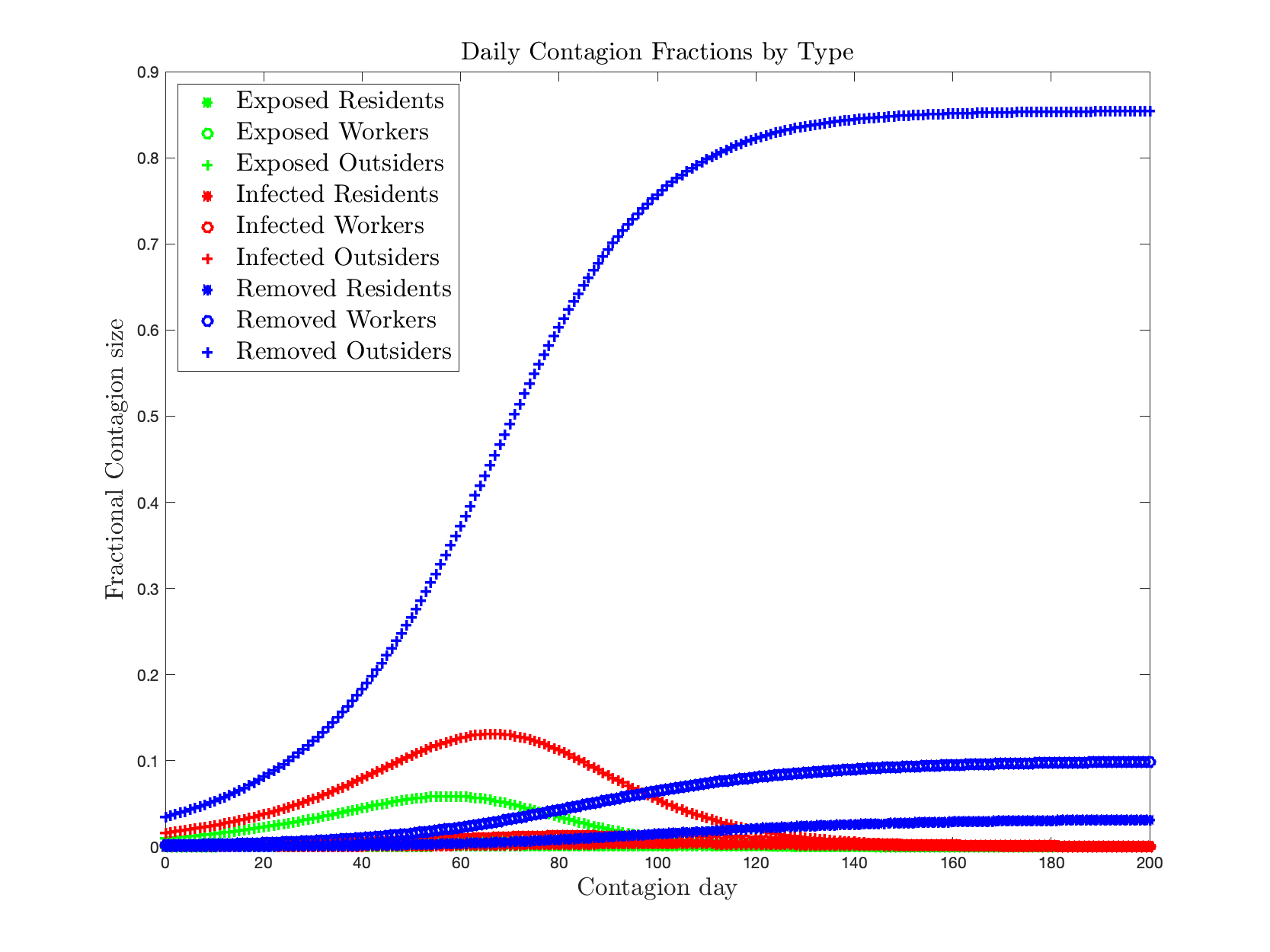}
\caption{Fractional contagion size by type and compartment in the Senior's Residential Centre Model of Section 6. Top Left: Benchmark strategy; top right: Strategy A; bottom left: Strategy B; bottom right: both Strategies A and B}\label{Fig1}

\end{figure}

What further policy improvements implemented by the management might lead to a better result? The remaining plots in Figure \ref{Fig1} show the results for several combinations of policy interventions. Strategy A is to improve internal hygiene by quarantining all residents and dramatically reducing contacts between workers: $\lambda(1,1)$ changes from $4$ to $0.5$ and $\lambda(2,2)$ changes from $5$ to $1$. Strategy B is to dramatically reduce the connectivity between the centre and the outside: $\lambda(2,3)$ changes from $4$ to $0.5$. We observe Strategy A manages to reduce the contagion to about 10\% of the residents, but allows a continual reintroduction of infection from outside. Strategy B fails outright: reducing the connections to outsiders simply delays the onset of contagion within the centre by about 30 days. However, the combination of both strategies A and B led to a success in keeping 97\% of the residents healthy.

These policy interventions target the social connectivity in the network through social distancing and quarantine. Another important channel would be to reduce the mean viral exposures entering in the exposure PDFs, by measures such as encouraging more cleanliness and the use of masks. Yet another channel is to improve individual immunity buffers by vaccination or other health improvements. 


 Large $N$ networks typically exhibit ``resilient'' states that are intrinsically resistant to  contagion and ``susceptible'' states that amplify any introduced infection. Moreover they can be made to  transition discontinuously from a resilient state to a susceptible state by varying a key parameter, such as the infective contact parameter $z$ that measures the degree of social distancing in the network. Figure \ref{Fig2} shows the long-time values of the removed fractions, as functions of $z$. One sees the remarkable transition from resilient to susceptible at a critical value $z^*\sim 0.055$. This single graph shows clearly the general principle that any contagion can be prevented at the outset by sufficiently strong restrictions on social interactions. 
 
 \begin{figure}[ht]
\centering
\includegraphics[width=0.65\columnwidth]{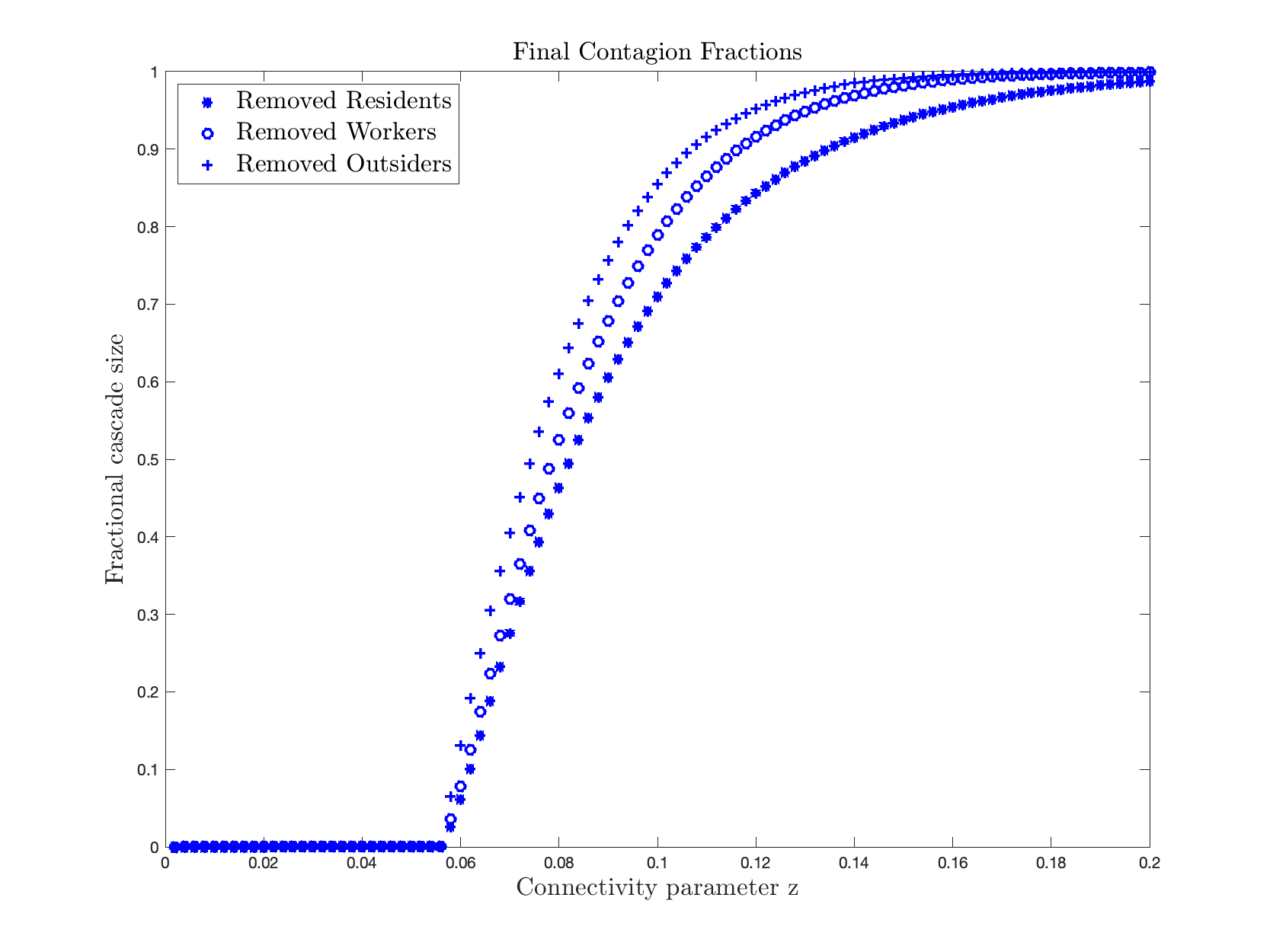}\caption{Final removed fractions as a function of $z$ for the benchmark Senior's Residential Centre Model of Section 8.}\label{Fig2}

\end{figure}

%
%
%

\section{Discussion}

The primary intention of this paper is to set out the fundamental assumptions and their consequences, for a novel network approach to epidemic modelling in very heterogeneous settings. To keep focussed on this aim, many potential examples and  avenues of inquiry have not been explored here. Instead, let us end this paper by discussing briefly how the novel features of the IRSN framework can be used in different fields to improve our understanding of COVID-19.
\begin{enumerate}
  \item {\bf To inform health policy:\ } A wide variety of scenarios such as the spread of disease between communities  can be explored within this framework. Once the IRSN model has been fully specified and calibrated to a real world setting, the analytical algorithm is straightforward to run. Since the IRSN starts with very different assumptions to standard tools such as the compartment ODE models, the exercise of implementing the IRSN forces policy makers to think in a different way about epidemics. This kind of modelling exercise will lead to more robust and reliable decisions that depend less on specific underlying assumptions. 
   \item {\bf To inform health research:\ } The IRSN framework can be extended to encompass a broad set of  characteristics that describe the immunology of COVID, the behaviour of human society and the effect of public health policy. Many details of the disease, particularly those connected with the threshold picture of viral transmission, are still inadequately understood. The IRSN can be used by researchers to study which gaps in data and knowledge may be leading to the greatest uncertainties in projections. This will suggest where scarce research funding should be best deployed. 
\item {\bf To inform network science:\ } The large $N$ analytical shortcut used in this paper is well known in network science, but has not yet been used in disease modelling. The meaning, accuracy and limitations of this shortcut will be of interest to other network modellers. A particularly interesting subtlety worthy of further research is to better understand the selection biases inherent in cascade dynamics on stochastic networks.   As well, the IRSN setting, being very analogous to the Inhomogeneous Random Financial Networks introduced in \cite{Hurd19a},  should be deployable in many other network applications. Modellers will observe that computational complexity is determined by the parameters $\Ndft,M$, and it is of interest to explore the tradeoffs when allocating  computational resources to a complex modelling problem at hand.   \item {\bf As a teaching tool: \ } Being easy to run on MATLAB or Python, the IRSN framework can be used in higher education as a learning and visualization tool that focusses on mathematical modelling assumptions for epidemics and their consequences. More broadly, these tools may be helpful in fostering public awareness of the most important societal issues, notably the effectiveness of targeted social distancing, that successful COVID health policy must address.\end{enumerate}

\section{Acknowledgements} This project was funded by the Natural Sciences and Engineering Research Council of Canada and the McMaster University COVID-19 Research Fund. The author is grateful to Hassan Chehaitli, Vladimir Nosov, Weijie Pang and Irshaad Oozeer for extensive discussions during the writing of this paper. 
\bibliographystyle{plainnat}


\end{document}